\documentclass[prl,reprint,superscriptaddress]{revtex4-2}
\usepackage[utf8]{inputenc}

\usepackage[dvipsnames]{xcolor}
\usepackage{amsmath}
\usepackage{graphicx}
\usepackage{float}
\usepackage{braket, amssymb}
\usepackage{rotating}
\usepackage{amsmath}
\usepackage{hyperref}

\usepackage{amsthm}

\usepackage[T1]{fontenc}

\newtheorem*{criteria}{Criterion}
\newtheorem*{lemma}{Lemma}

\newcommand{\Tr}{{\rm Tr}}
\begin{document}
\title{Construction of Non-Hermitian Parent Hamiltonian from Matrix Product States}
\author{Ruohan Shen}
\thanks{These authors contributed equally.}
\affiliation{State Key Laboratory of Low Dimensional Quantum Physics and Department of Physics, Tsinghua University, Beijing 100084, China}

\author{Yuchen Guo}
\thanks{These authors contributed equally.}
\affiliation{State Key Laboratory of Low Dimensional Quantum Physics and Department of Physics, Tsinghua University, Beijing 100084, China}
\author{Shuo Yang}
\email{shuoyang@tsinghua.edu.cn}
\affiliation{State Key Laboratory of Low Dimensional Quantum Physics and Department of Physics, Tsinghua University, Beijing 100084, China}
\affiliation{Frontier Science Center for Quantum Information, Beijing 100084, China}
\affiliation{Hefei National Laboratory, Hefei 230088, China}

\date{\today}

\begin{abstract}
    There are various research strategies used for non-Hermitian systems, which typically involve introducing non-Hermitian terms to pre-existing Hermitian Hamiltonians. 
    It can be challenging to directly design non-Hermitian many-body models that exhibit unique features not found in Hermitian systems.
    In this Letter, we propose a new method to construct non-Hermitian many-body systems by generalizing the parent Hamiltonian method into non-Hermitian regimes.
    This allows us to build a local Hamiltonian using given matrix product states as its left and right ground states. 
    We demonstrate this method by constructing a non-Hermitian spin-$1$ model from the asymmetric Affleck-Kennedy-Lieb-Tasaki (AKLT) state, which preserves both chiral order and symmetry-protected topological order.
    Our approach opens up a new paradigm for systematically constructing and studying non-Hermitian many-body systems, providing guiding principles to explore new properties and phenomena in non-Hermitian physics.
\end{abstract}
\maketitle

\emph{Introduction.\textemdash}
Non-Hermitian physics has attracted much attention both theoretically~\cite{Konotop2016, Leykam2017, Gong2018, Kawabata2019, Ashida2020, Bergholtz2021, Yamamoto2022, Ding2022, Guo2022, Chen2022} and experimentally~\cite{Guo2009, Schindler2011, Zeuner2015, El2018, Hokmabadi2019, Zhang2022A, Gu2022} for describing open systems~\cite{Vega2017}, such as photonics~\cite{Takata2018} and acoustics~\cite{Zhang2021A, Wen2022} with gain and loss, as well as quasi-particles in interacting or disordered systems~\cite{Kozii2017, Shen2018}.
It has also revealed non-trivial properties that have no Hermitian counterpart~\cite{Heiss1990, Yao2018, Kunst2018, Xiong2018, Matsumoto2020, Borgnia2020}.

However, many recent studies revealing non-trivial properties of non-Hermitian systems have focused on the single-particle picture~\cite{Hatano1996, Zeng2020, Esaki2011}.
One reason for this is that many powerful numerical methods for Hermitian quantum many-body models, such as density matrix renormalization group (DMRG)~\cite{White1992, Schollwoeck2005} and quantum Monte Carlo (QMC)~\cite{Gubernatis2016}, cannot be directly applied to non-Hermitian systems.
Some modified algorithms also suffer from unstable convergence~\cite{Chan2005, Rotureau2006} and incapability near exceptional points~\cite{Huang2011A, Huang2011B, Zhang2020}.

Therefore, it is interesting to consider the opposite question: can we construct a non-Hermitian Hamiltonian from a pair of easily engineered states that preserve desired properties, rather than having to extract various properties from a given Hamiltonian?
In Hermitian systems, this task can be achieved using the parent Hamiltonian method~\cite{Affleck1987, PerezGarcia2007}, which allows for the construction of a local, gapped Hamiltonian whose ground state is represented by a matrix product state (MPS)~\cite{Verstraete2008, Schollwock2011, Schuch2013, Orus2014, Cirac2021}.
However, this method cannot be directly applied to non-Hermitian systems.

In this Letter, we present a method for constructing non-Hermitian parent Hamiltonians (nH-PHs) by generalizing the conventional Hermitian approach.
We provide criteria for states that can be used to establish an nH-PH and derive the explicit form of the Hamiltonian.
As an example, we construct a non-Hermitian model from asymmetric AKLT states~\cite{Maekawa2022} and examine its physical properties in the thermodynamic limit using the generalized infinite time-evolving block decimation (iTEBD) method~\cite{Hastings2009, SeeSM}.
We find that the model has two non-trivial orders: chiral order detected by a local order parameter and symmetry-protected topological (SPT) order~\cite{Gu2009, Chen2010, Chen2013} detected by a string order parameter~\cite{PerezGarcia2008}.

\emph{Non-Hermitian Parent Hamiltonian.\textemdash}
The expectation value of any observable $\braket{\hat{O}}$ for a general non-Hermitian system can be evaluated in different ways~\cite{Brody2013, Lee2020, Grimaudo2020, Grimaldi2021}.
Here we choose the formalism discussed in~\cite{Brody2013} to calculate the expectation as
\begin{align}
    \braket{\hat{O}}_{LR} = \braket{L|\hat{O}|R}/\braket{L|R},
    \label{ObservableLR}
\end{align}
which has a clear geometric interpretation~\cite{Ju2019}.
Here $\ket{R}$ and $\ket{L}$ are the ground states of $H$ and $H^{\dagger}$ respectively, which are defined as the eigenstates with the lowest real parts of the eigenvalues.
As a result, many more novel properties emerge in non-Hermitian systems since we have more degrees of freedom in choosing $\bra{L}$ independent of $\ket{R}$ than in the Hermitian case, where expectation values are evaluated under $\braket{\hat{O}}_{RR} = \braket{R|\hat{O}|R}/\braket{R|R}$.
A natural question arises: can such a system be constructed, i.e., can we find a non-Hermitian Hamiltonian that has the given $\bra{L}$ and $\ket{R}$ as its corresponding ground states?
In the following, we answer this question for MPSs, which satisfy the entanglement area law and can describe ground states of one-dimensional (1D) local and gapped Hamiltonians~\cite{Hastings2007, Eisert2010}.
The same argument can be easily extended to higher dimensions.

\begin{figure}[tbp]
    \centering
    \includegraphics[width=1.0\columnwidth]{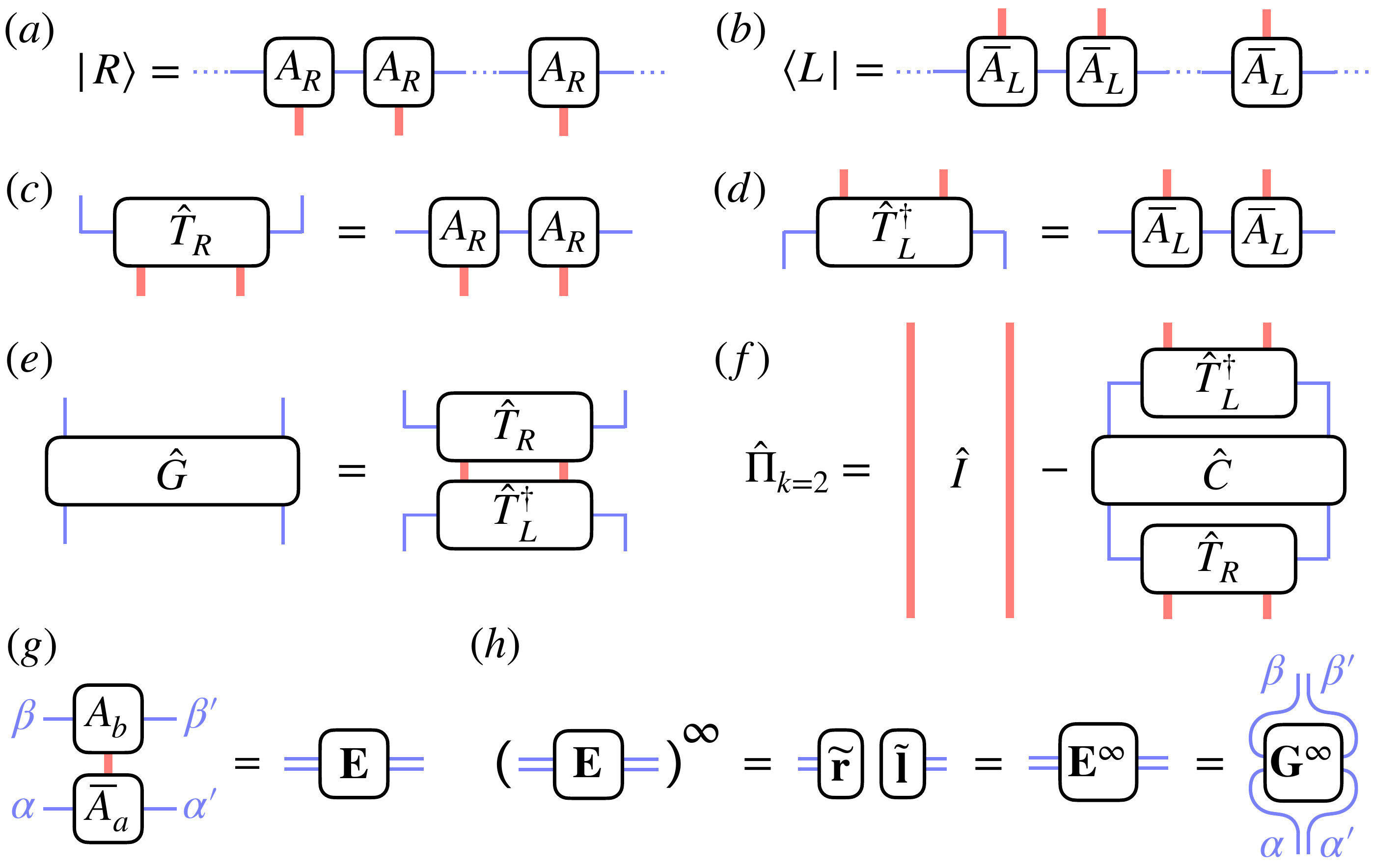}
    \caption{Construction of non-Hermitian parent Hamiltonian for $k=2$.
             (a-b) Translation-invariant MPS.
             (c-d) Local tensors $\hat{T}_R = \ket{\mathbf{p}} \mathbf{T}_R (\mathbf{r}|$ and $\hat{T}_L^{\dagger} = |\mathbf{l}) \mathbf{T}_L^{\dagger}\bra{\mathbf{p}}$.
             (e) The metric operator $\hat{G} = \hat{T}_L^{\dagger}\hat{T}_R$.
             (f) The local projector $\hat{\Pi} = \hat{I} - \hat{T}_R\hat{C}\hat{T}_L^{\dagger}$.
             (g) The transfer matrix $\mathbf{E}$ constructed from tensors $A_b$ and $A_a^{*}$.
             (h) The RG fixed point $\mathbf{E}^{\infty}$ and the corresponding $\mathbf{G}^{\infty}$.
    }
    \label{Projector}
\end{figure}

In this context, we consider 1D translation-invariant (TI) and injective MPS~\cite{PerezGarcia2007} $\ket{R}$ and $\ket{L}$ shown in Fig.~\ref{Projector}(a,b), written as 
\begin{align}
    \ket{R(L)} = \sum\limits_{i_1,\dots,i_N} \mathrm{Tr}\left[A_{R(L)}^{[i_1]}\dots A_{R(L)}^{[i_N]}\right] \ket{i_1,\dots,i_N}
\end{align}
with the same virtual bond dimension $D$ and physical bond dimension $d$.
A detailed explanation of MPS is shown in Supplemental Material~\cite{SeeSM}.
As shown in Fig.~\ref{Projector}(c,d), tensors after contracting $k$ neighboring sites can be regarded as maps from virtual to physical degrees of freedom, i.e.,
$
    \hat{T}_R = \ket{\mathbf{p}} \mathbf{T}_R (\mathbf{r}|\text{ and }
    \hat{T}_L^{\dagger} = |\mathbf{l}) \mathbf{T}_L^{\dagger} \bra{\mathbf{p}}.
$
Here, $\ket{\mathbf{p}}$ is the collective physical basis and $|\mathbf{r})$ ($(\mathbf{l}|$) is the collective virtual basis of $\ket{R}$ ($\bra{L}$).
$\mathbf{T}_R$ and $\mathbf{T}_L$ are coefficient matrices.
The local support spaces $\mathcal{H}_R$ and $\mathcal{H}_L$ are the images of $\hat{T}_R$ and $\hat{T}_L$, respectively.
With the injectivity condition, we can choose a large enough $k$ such that $d^k\geq D^2$ and $\dim{\mathcal{H}_R} = \dim{\mathcal{H}_L} = D^2$~\cite{PerezGarcia2007}.

We aim to find a local Hamiltonian in the form $\hat{H}= \sum_i \hat{\Pi_i}$, where each $\hat{\Pi_i}=\hat{I}-\hat{P}_i$ acts on $k$ local sites and ensures that $\bra{L}$ and $\ket{R}$ are zero-energy modes.
Here $\hat{P}_i$ is a projector with $\hat{P}_i^2=\hat{P}_i$.
In other words, we require $\hat{P}_i$ and $\hat{P}_i^{\dagger}$ to be projectors onto $\mathcal{H}_R$ and $\mathcal{H}_L$, respectively, 
\begin{align}
    \hat{P}_i\hat{T}_R = \hat{T}_R,\qquad \hat{T}_L^{\dagger}\hat{P}_i = \hat{T}_L^{\dagger}.\label{ProjectorCondition}
\end{align}
Meanwhile, we require ${\rm rank} \hat{P} = {\rm rank}\hat{P}^{\dagger} = D^2$.
Therefore, the most general form of a projector can be written as (the site index $i$ has been omitted for simplicity)
\begin{align}
    \hat{P} = \ket{\mathbf{p}} \mathbf{T}_R \mathbf{C} \mathbf{T}_L^{\dagger} \bra{\mathbf{p}} = \hat{T}_R\hat{C}\hat{T}_L^{\dagger},
    \label{GeneralFormOfProjector}
\end{align}
where $\mathbf{C}$ is a $D^2\times D^2$ matrix and $\hat{C} = |\mathbf{r}) \mathbf{C} (\mathbf{l}|$ is its operator form.
The matrix elements are determined by Eq.~\eqref{ProjectorCondition}, $\hat{T}_R\hat{C}\hat{T}_L^{\dagger}\hat{T}_R = \hat{T}_R $.
As $\mathrm{rank}\,\hat{T}_R = \mathrm{rank}\,\hat{T}_L^{\dagger} = D^2$, this gives 
\begin{align}
\hat{C} = (\hat{T}_L^{\dagger} \hat{T}_R)^{-1} \equiv \hat{G}^{-1}.
\end{align}
Therefore, the metric operator $\hat{G} = |\mathbf{l}) \mathbf{G} (\mathbf{r}|$ shown in Fig.~\ref{Projector}(e) must be invertible and can fully determine the $k$-local Hamiltonian shown in Fig.~\ref{Projector}(f)
\begin{align}
    \hat{\Pi} = \hat{I} - \hat{T}_R\hat{C}\hat{T}_L^{\dagger} 
    = \ket{\mathbf{p}}\mathbf{I}\bra{\mathbf{p}} - \ket{\mathbf{p}} \mathbf{T}_R \mathbf{G}^{-1} \mathbf{T}_L^{\dagger} \bra{\mathbf{p}}.\label{ProjectorFinal}
\end{align}
In Fig.~S1 in Supplemental Material~\cite{SeeSM}, we verify Eq.~\eqref{ProjectorCondition} in a straightforward way.

We note that $\hat{G}^{-1}$ is simply the operator used to bi-orthogonalize $\mathcal{H}_R$ and $\mathcal{H}_L$.
In other words, if we perform the transformations $\mathbf{T}_R\rightarrow \mathbf{T}_R\mathbf{G}^{-1}$ and $\mathbf{T}_L\rightarrow\mathbf{T}_L$, $\hat{T}_L$ and $\hat{T}_R$ become orthogonal operators
\begin{align}
    \hat{T}_L^{'\dagger} \hat{T}_R^{'} = |\mathbf{l}) \mathbf{T}_L^{\dagger} \mathbf{T}_R \mathbf{G}^{-1} (\mathbf{r}| = \hat{I}.
\end{align}
On the other hand, the ability to perform bi-orthogonalization guarantees the existence of nH-PH, as proved in Supplemental Material~\cite{SeeSM}.
As a specific example, when we set $\hat{T}_L = \hat{T}_R$, our method reproduces the conventional Hermitian projector, but with a clearer physical interpretation.

By referencing the proof in Ref.~\cite{PerezGarcia2007}, we see that the given right (left) state is guaranteed to be the unique zero-energy eigenstate of $\hat{H}$ ($\hat{H}^{\dagger}$) by construction.
However, the zero mode is not necessarily the ground state due to non-Hermiticity.
Specifically, we have $\braket{\hat{H}} = \bra{\psi}\sum_i \hat{h}_i\ket{\psi} = \sum_i \bra{\psi}\hat{h}_i \ket{\psi} \geq \sum_i E_{0,i}$ for a Hermitian parent Hamiltonian, where $E_{0, i}$ is the ground state energy of the local term $\hat{h}_i$.
Thus, the total system energy is bound by the local ground-state energies.
However, this inequality no longer holds in the non-Hermitian regime. 
$\Re({\bra{\psi} \hat{h}_i \ket{\psi}})$ can be even smaller than $\Re{(E_{0, i})}$ (which equals to $0$ in our nH-PH), implying the existence of a negative energy eigenstate.
As a consequence, the bound on the total energy disappears, and the common ground state of local projectors is not necessarily the global ground state.
This phenomenon often occurs when non-Hermitian effects are significant but can be reduced by increasing the interaction length $k$, as demonstrated in the following example.

\emph{$\mathcal{PT}$-symmetry.\textemdash}
In the following, we consider non-Hermitian systems with $\mathcal{PT}$-symmetry.
These systems are particularly noteworthy because their spectra only contain real numbers or conjugate pairs~\cite{Bender1998, Bender2002, Konotop2016}, and they can be easily implemented and maintained in our nH-PH by designing $\ket{R}$ and $\bra{L}$.

We construct a pesudo-Hermitian Hamiltonian~\cite{Ashida2020} satisfying that $\hat{\mathcal{P}} \hat{H}^{\dagger} \hat{\mathcal{P}}^{-1}=\hat{H}$ and $\hat{\mathcal{T}}\hat{H}^{\dagger}\hat{\mathcal{T}}^{-1}=\hat{H}$, where $\hat{\mathcal{P}}$ and $\hat{\mathcal{T}}=e^{i\pi \hat{S}_y}\hat{K}$ are the parity symmetry and time-reversal symmetry operators, respectively.
This results in the $\mathcal{PT}$ joint symmetry $\hat{\mathcal{P}}\hat{\mathcal{T}} \hat{H} (\hat{\mathcal{P}}\hat{\mathcal{T}})^{-1}=\hat{H}$.
Meanwhile, the above condition requires that the ground state of $H$ and $H^{\dagger}$ be connected by similar transformations $\hat{\mathcal{P}}$ or $\hat{\mathcal{T}}$.
To construct such a non-Hermitian Hamiltonian, we need a TI MPS that does not preserve $\hat{\mathcal{P}}$ or $\hat{\mathcal{T}}$ symmetry itself, but satisfies the joint symmetry condition $\sum_j (e^{-i\pi \hat{S}_y})_{i,j} (A^{[j]})^{*} \propto M^{-1}(A^{[i]})^{\rm T}M$.
Here $\hat{\mathcal{P}}$ is realized by exchanging two virtual indices for a TI MPS, and $M$ is an arbitrary gauge on the virtual indices of the right ground state $\ket{R}$.
The left ground state is chosen as $\ket{L}=\hat{\mathcal{P}}\ket{R}$, whose tensors are given by ${A^{\prime[i]}} = (A^{[i]})^{\rm T}$.

\emph{Asymmetric AKLT model\textemdash}
We use the asymmetric AKLT state as the right ground state $\ket{R} = \ket{\Phi_\mu}$~\cite{Maekawa2022}, which satisfies the aforementioned conditions. 
This state can be represented by an MPS with the following non-zero elements
\begin{align}\begin{aligned}
    &A_{\mu,\downarrow \uparrow}^{[1]} = -\sqrt{\mu},\,
    &&A_{\mu,\uparrow \downarrow}^{[-1]} = \sqrt{\mu}, \\
    &A_{\mu,\uparrow \uparrow}^{[0]} = 1/\sqrt{2},\,
    &&A_{\mu,\downarrow \downarrow}^{[0]} = -\mu/\sqrt{2}.
    \label{ComponentsAmu}
\end{aligned}\end{align}
Its entanglement structure is similar to that of the AKLT state, with an asymmetric underlying valence bond $\ket{\uparrow\downarrow}-\mu\ket{\downarrow\uparrow}$ tending toward one side. 
It is worth noting that $\ket{L} \propto \ket{\Phi_{1/\mu}}$ since their local tensors are related by a gauge transformation on virtual indices $(A_{\mu}^{[i]})^{\rm T} = - \mu \hat{\sigma}_y A_{1/\mu}^{[i]} \hat{\sigma}_y$~\cite{SeeSM}.

We first focus on the region $\mu\in[0, 1]$ for simplicity, and will reveal the reason later.
To calculate the expectation value of any observable $\braket{O}_{LR}$ in the thermodynamic limit, we need to evaluate the composed transfer matrix~\cite{Schollwock2011, Schuch2013} defined in Fig.~\ref{Projector}(g) using $A_{b}=A_\mu$ and $A_{a}=A_{\mu}^{\rm T}$.
On the basis $\{\uparrow \uparrow, \uparrow \downarrow, \downarrow \uparrow, \downarrow \downarrow \}$, we obtain
\begin{align}
    \mathbf{E}_{\mu,(\alpha \beta, \alpha' \beta')}
    = \left(\frac{1}{2}\right) \oplus
    \left(
        \begin{array}{cc}
            -\frac{\mu}{2} & \mu\\
            \mu & -\frac{\mu}{2}
        \end{array}
    \right)
    \oplus \left(\frac{\mu^2}{2}\right),
    \label{TransferMatrix}
\end{align}
whose eigenvalues are $\{\frac{1}{2}, -\frac{3\mu}{2}, \frac{\mu}{2}, \frac{\mu ^2}{2} \}$~\cite{SeeSM}.
At $\mu = \frac{1}{3}$, there is a `level crossing' transition for the dominant eigenvector of $\mathbf{E}_{\mu}$.
We will study this transition from the renormalization group (RG) perspective and conclude that it is a unique phenomenon that can only occur in non-Hermitian systems.

Implementation of RG aims to remove short-range entanglement and study long-range patterns.
This can be achieved from the fixed point of $\mathbf{E}$ via grouping infinite local tensors~\cite{Chen2010}, i.e., $\mathbf{E}^{\infty}=\mathrm{lim}_{k\rightarrow \infty} (\mathbf{E}/\lambda)^k$  with $\lambda$ being the dominant eigenvalue.
When $\mu>\frac{1}{3}$, the fixed point transfer matrix~\cite{SeeSM}
\begin{align}
    \mathbf{E}_{\mu,(\alpha \beta, \alpha' \beta')}^{\infty}(\mu>\frac{1}{3})
    = \frac{1}{2} \left(0\right) \oplus
    \left(
        \begin{array}{cc}
            1 & -1\\
            -1 & 1
        \end{array}
    \right)
    \oplus \left(0\right),
\end{align}
is the same as that of the conventional AKLT state up to a gauge, indicating that the non-Hermitian system is in the same AKLT phase for $\mu>\frac{1}{3}$.
On the contrary, $\mathbf{E}_{(\alpha \beta, \alpha' \beta')}^{\infty}(\mu<\frac{1}{3}) = \rm{Diag}\{1, 0, 0, 0\}$ is equivalent to a transfer matrix constructed from two product states, where any local observable would have a trivial expectation value.
Therefore, there is a quantum phase transition from the AKLT phase to the trivial phase at $\mu_{c} = \frac{1}{3}$, which can be detected by a chiral order parameter that will be introduced later.

At the same time, the corresponding metric matrix $\mathbf{G}_{(\alpha \alpha', \beta \beta')}^{\infty}(\mu<\frac{1}{3}) = \rm{Diag}\{1, 0, 0, 0\}$ shown in Fig.~\ref{Projector}(h) is not invertible, implying that the ability to bi-orthogonalize the local Hilbert spaces $\mathcal{H}_R^{\infty}$ and $\mathcal{H}_L^{\infty}$ will be destroyed during the RG process.
Thus, it is impossible to create a projector-form nH-PH for $k\rightarrow\infty$, even if $\mathbf{G}$ is invertible for finite $k$.
In summary, this new kind of phase transition without a Hermitian counterpart originates from the mismatch between the left and right ground states at the RG fixed point.

\emph{Chiral order and SPT order.\textemdash}
The asymmetric underlying valence bonds $\ket{\uparrow\downarrow}-\mu\ket{\downarrow\uparrow}$ and $-\mu\ket{\uparrow\downarrow}+\ket{\downarrow\uparrow}$ tend in opposite directions in $\ket{R}$ and $\ket{L}$, implying an interesting chiral property.
To detect this chiral order, two non-Hermitian order parameters are introduced $\hat{O}_{\rm left}=\frac{1}{2}\hat{S}_i^{+}\hat{S}_{i+1}^{-}$ and $\hat{O}_{\rm right}=\frac{1}{2}\hat{S}_i^{-}\hat{S}_{i+1}^{+}$. The chiral order parameter is then defined as $\hat{O}_{\rm chiral} = \hat{O}_{\rm right}-\hat{O}_{\rm left}$.
As a comparison, we also consider $\hat{O}_{\rm AF} = \hat{S}_i^{z}\hat{S}_{i+1}^{z}$, which is commonly adopted to detect the conventional anti-ferromagnetic order~\cite{SeeSM}.
The results for the non-Hermitian case are shown in Fig.~\ref{ChiralOrderParameter}(a).
For $\frac{1}{3}<\mu<3$, we obtain
\begin{align}
    \braket{\hat{O}_{\rm AF}} = -\frac{4}{9},\,
    \braket{\hat{O}_{\rm left}} = -\frac{4\mu}{9},\,
    \braket{\hat{O}_{\rm right}} = -\frac{4}{9\mu}.
\end{align}
At the AKLT point $\mu=1$, the state is isotropic with $\braket{\hat{O}_{\rm AF}}=\braket{\hat{O}_{\rm left}}=\braket{\hat{O}_{\rm right}}$.
Meanwhile, ${\rm sgn}\braket{\hat{O}_{\rm chiral}}$ changes when $\mu$ passes by $1$, demonstrating the chiral property of different directions.
In contrast, if we choose $\ket{L} = \ket{R}$ in Fig.~\ref{ChiralOrderParameter}(b), $\braket{\hat{O}_{\rm chiral}} = 0$ for all values of $\mu$.
This is because the chiral order parameter $\hat{O}_{\rm chiral}$ is anti-Hermitian, meaning that $\Re({\bra{\psi}\hat{O}_{\rm chiral}\ket{\psi}})=0$ for any $\ket{\psi}$.
As a result, such non-trivial chiral order cannot be realized in Hermitian systems.

\begin{figure}[tbp]
    \centering
    \includegraphics[width=\linewidth]{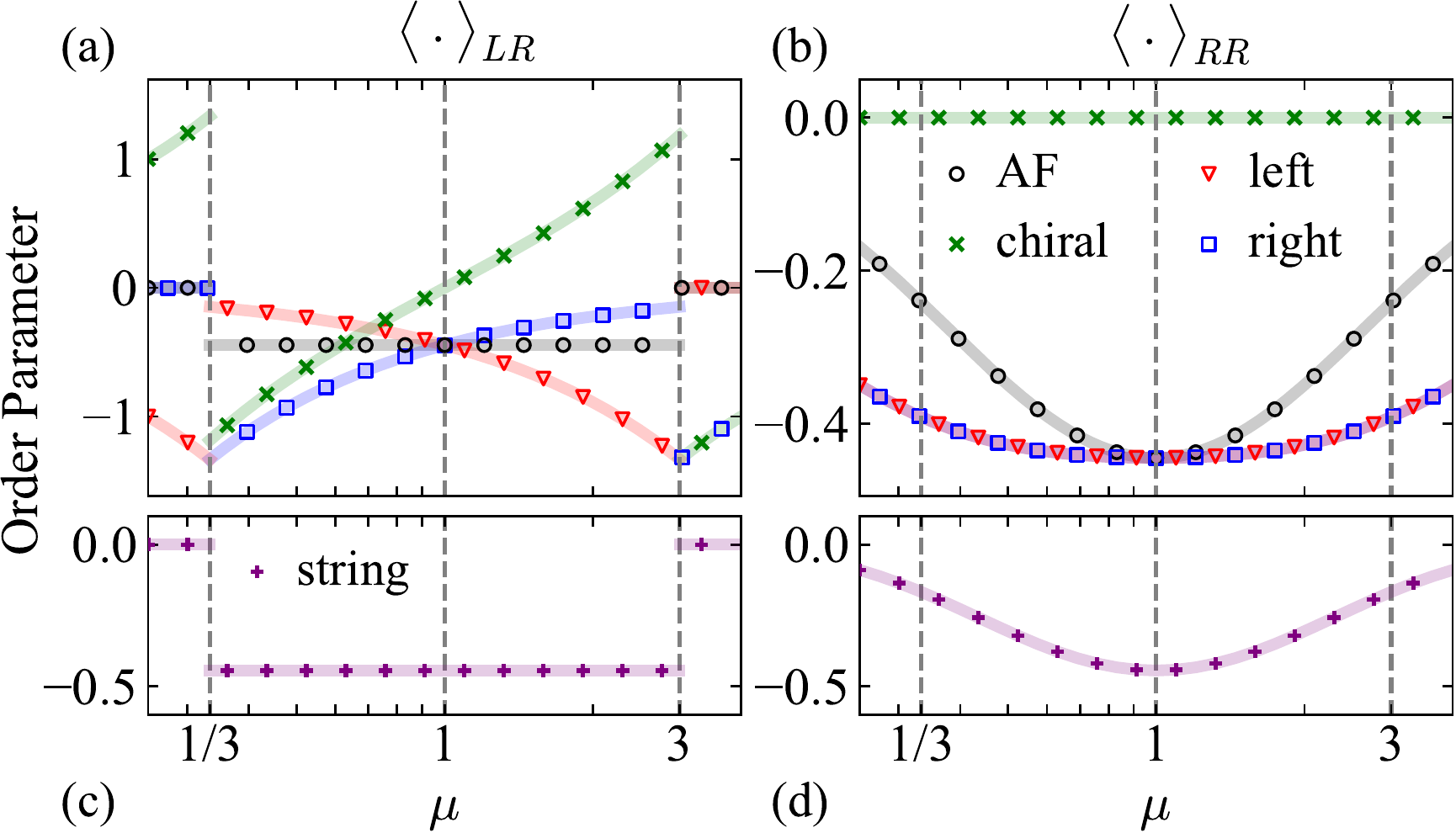}
    \caption{Order parameters evaluated under different $\mu$.
    The $x$-axis is presented in a log scale.
    (a)(c) The expectation values of chiral and string order parameters for non-Hermitian systems.
    (b)(d) The same for Hermitian systems.
    }
    \label{ChiralOrderParameter}
\end{figure}

In addition, there is a duality $\mu\sim \frac{1}{\mu}$ for $H\sim H^{\dagger}$, which is induced by the parity operation, i.e., $\ket{\uparrow\downarrow}-\mu\ket{\downarrow\uparrow} \rightarrow -\mu\ket{\uparrow\downarrow}+\ket{\downarrow\uparrow} = -\mu(\ket{\uparrow\downarrow}-\frac{1}{\mu}\ket{\downarrow\uparrow})$.
Since $H$ and $H^{\dagger}$ share the same energy spectrum, this relation directly gives the isotropic point $\mu=1$ and explains why the transitions from non-trivial to trivial systems occur in pairs at $\mu=\frac{1}{3}$ and $\mu=3$.
For the same reason, the chiral order parameter in Fig.~\ref{ChiralOrderParameter}(a) is centrosymmetric.

Our system also exhibits non-trivial SPT order.
We use $\hat{O}_{\rm string}(i, j)=\hat{S}_i^z(\prod_{k=i+1}^{j-1}e^{i\pi \hat{S}_k^z})\hat{S}_j^z$, which was previously adopted for the conventional AKLT state~\cite{Nijs1989,PerezGarcia2008}, to detect the SPT order in our non-Hermitian system.
Its expectation value can be calculated analytically in the thermodynamic limit~\cite{Maekawa2022, SeeSM}, and the result is shown in Fig.~\ref{ChiralOrderParameter}(c).
For $\frac{1}{3}<\mu<3$, the system preserves perfect non-decaying string order $\braket{\hat{O}_{\rm string}}=-\frac{4}{9}$ for any string length, indicating that it is in the same SPT phase as the conventional AKLT model.
Nevertheless, the string order vanishes for $\mu<\frac{1}{3}$ and $\mu>3$, showing that it is similar to a trivial product state.
This is consistent with previous discussions.
In contrast, for the Hermitian system shown in Fig.~\ref{ChiralOrderParameter}(d), the string order parameter also saturates to a non-zero value for all $\mu$~\cite{SeeSM}, but the value becomes smaller as $\mu$ deviates from the AKLT point.

\emph{Parent Hamiltonian.\textemdash}
Here we explicitly construct a TI non-Hermitian Hamiltonian to realize the aforementioned chiral and SPT orders with $k=2$, i.e., with only nearest-neighbor interactions~\cite{SeeSM}
\begin{align}
    &\hat{\Pi}(\mu)_{i} 
    = \frac{5}{12}\left(\frac{\mu}{2} \hat{S}_{i}^{-}\hat{S}_{i+1}^{+} + \frac{1}{2\mu}\hat{S}_{i}^{+}\hat{S}_{i+1}^{-} + \hat{S}_{i}^{z}\hat{S}_{i+1}^{z}\right)+\frac{2}{3} \nonumber \\
    &+ \frac{1}{6}\left(\frac{\mu^2}{4}{\hat{S}_{i}^{-2}}{\hat{S}_{i+1}^{+2}}+\frac{1}{4\mu^2}{\hat{S}_{i}^{+2}}{\hat{S}_{i+1}^{-2}}-{\hat{S}_{i}^{z2}}-{\hat{S}_{i+1}^{z2}}\right) \nonumber \\
    &+ \frac{1}{24}\left(\mu \hat{S}_{i}^{-z}\hat{S}_{i+1}^{+z} + \frac{1}{\mu}\hat{S}_{i}^{+z}\hat{S}_{i+1}^{-z}\right)
    +\frac{1}{4}{\hat{S}_{i}^{z2}}{\hat{S}_{i+1}^{z2}},
    \label{k=2nHPH}
\end{align}
where $\hat{S}^{\pm z}=\hat{S}^{\pm}\hat{S}^z+\hat{S}^z\hat{S}^{\pm}$. 
It is obvious that $\hat{H}(\mu)=\sum_i \hat{\Pi}(\mu)_i$ does not preserve either $\hat{\mathcal{P}}$ (exchanging site $i$ and $i+1$) or $\hat{\mathcal{T}}$ ($\hat{S}^z\rightarrow-\hat{S}^z$, $\hat{S}^+\rightarrow-\hat{S}^-$, $\hat{S}^-\rightarrow-\hat{S}^+$) individually, but remains unchanged when $\hat{\mathcal{P}}$ and $\hat{\mathcal{T}}$ are combined.

We use exact diagonalization (ED) to investigate the energy spectrum for small systems and find that the spectrum for open boundary condition (OBC) is identical for all $\mu>0$~\cite{SeeSM}, with four-fold degenerate ground states as a characteristic property of SPT~\cite{Gu2009}.
We also calculate the spectrum under periodic boundary condition (PBC) and show that the Hamiltonian remains gapped for a wide range of $\mu$ via finite-size scaling to $N \rightarrow \infty$, as shown in Fig.~S5 in Supplemental Material~\cite{SeeSM}.

According to previous sections, the phase transitions occur at $\mu_{\rm c}=\frac{1}{3}$ and $\mu_{\rm c}=3$ for $k\rightarrow \infty$.
In this case, eigenvalues with negative real parts will not appear, and the invertibility of $\mathbf{G}^{\infty}$ is equivalent to the existence of an nH-PH with $\ket{\Phi_{\mu}}$ as its ground state.
When using finite $k$, the nH-PH $\hat{H}_{k}\left(\mu\right)$ is still well-defined, even when $\mu<\frac{1}{3}$ and $\mu>3$, but it does not have $\ket{\Phi_{\mu}}$ as its unique ground state in these regions.
Furthermore, the construction of nH-PH may have unfavorable consequences, such as level-crossing caused by the non-commutability of local projectors and non-Hermiticity, which shifts the critical points towards the intermediate phase for finite $k$. 

To detect phase transitions, we generalize the modified iTEBD method~\cite{Hastings2009} to analyze Hamiltonians with multi-site interactions~\cite{SeeSM}.
We find that $k=2$ is sufficient to identify chiral and string orders for a wide range of $\mu$ in the intermediate phase.
We evaluate the infidelity between the output state from iTEBD $\ket{\Psi_{\mu}}$ with $D = 12$ and the given asymmetric AKLT state $\ket{\Phi_{\mu}}$, which is defined as $\eta=1-\lim_{N\rightarrow\infty}\left|\braket{\Phi_{\mu}|\Psi_{\mu}}\right|^{1/N} = 1-|\lambda_{\Phi\Psi}|$ with normalization condition $\braket{\Phi_\mu | \Phi_\mu}=1$ and $\braket{\Psi_\mu | \Psi_\mu}=1$.
The results shown in Fig.~\ref{iTEBDresult}(a-b) indicate that the asymmetric AKLT state $\ket{\Phi_{\mu}}$ is indeed the ground state in the intermediate phase, but not for extreme values of $\mu$ near the regions $\mu<\frac{1}{3}$ and $\mu>3$, although it is always a zero mode by construction.
As we increase $k$, the critical point will converge to $\mu_{\rm c}$.
Using $k=3$ allows us to expand the region of nH-PH and brings the critical points much closer to $\mu_{\rm c}$.

\begin{figure}[tbp]
    \centering
    \includegraphics[width=1\linewidth]{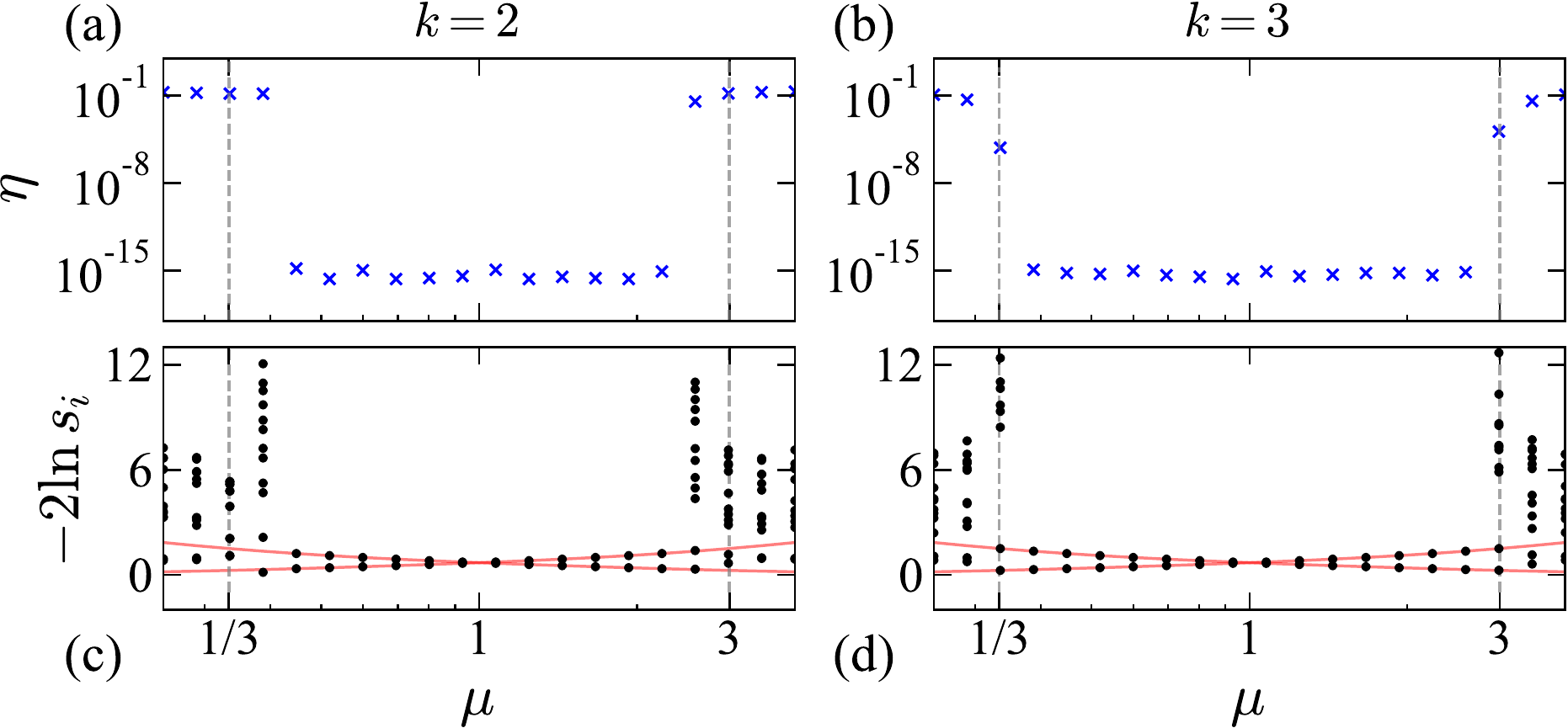}
    \caption{Calculated ground state $\ket{\Psi_{\mu}}$ of $H_k(\mu)$ using the multi-site iTEBD method with $D = 12$ for $k = 2$ and $k = 3$.
    (a-b) Infidelity between $\ket{\Psi_{\mu}}$ and $\ket{\Phi_{\mu}}$.
    (c-d) Entanglement spectrum of $\ket{\Psi_{\mu}}$ (black dots) and $\ket{\Phi_{\mu}}$ (red lines).
    }
    \label{iTEBDresult}
\end{figure}

Moreover, we investigate the entanglement spectrum of $\ket{\Psi_{\mu}}$ in Fig.~\ref{iTEBDresult}(c-d).
In the intermediate phase, the ground state $\ket{\Psi_{\mu}}$ has only two non-zero elements in the entanglement spectrum, consistent with that of $\ket{\Phi_{\mu}}$ shown in red curves. 
On the contrary, for extreme $\mu$, the algorithm cannot converge to a unique ground state and the entanglement spectrum is gapless.
Numerical simulations for $H_k^{\dagger}(\mu)$ whose ground state is expected to be $\ket{\Phi_{1/\mu}}$ are shown in Fig.~S7 in Supplemental Material~\cite{SeeSM}, where we obtain consistent results.

\emph{Conclusion.\textemdash}
In this Letter, we propose a general scheme to construct a non-Hermitian Hamiltonian from two different MPS $\bra{L}$ and $\ket{R}$ as left and right ground states.
As an example, we demonstrate how to create a non-Hermitian model from asymmetric AKLT states that perserves both chiral and SPT orders, and identify a phase transition with a new origin without a Hermitian counterpart.

Our approach changes the paradigm of non-Hermitian physics, from top-down to bottom-up.
We can now construct Hamiltonians with short-range interactions from states that preserve desired properties rather than extracting information from a given Hamiltonian.
Compared to the conventional Hermitian parent Hamiltonian, our method offers more possibilities as there are extra degrees of freedom in choosing two states instead of one.
It also establishes a duality between quantum states and Hamiltonians, liberates researchers from the constraints of specific systems, and provides a new perspective to study strongly correlated quantum many-body systems.
We believe that there is a broader world in strongly-correlated many-body systems in the non-Hermitian regime.

\begin{acknowledgments}
    We thank Ze-An Xu and Yanzhen Wang for helpful discussions.
    This work is supported by the National Natural Science Foundation of China (NSFC) (Grant No. 12174214 and No. 92065205), the National Key R\&D Program of China (Grant No. 2018YFA0306504), the Innovation Program for Quantum Science and Technology (Grant No. 2021ZD0302100), and the Tsinghua University Initiative Scientific Research Program.
\end{acknowledgments}

\bibliography{ref}
\newpage
\appendix
\onecolumngrid
\renewcommand{\thesection}{S-\arabic{section}} \renewcommand{\theequation}{S%
\arabic{equation}} \setcounter{equation}{0} \renewcommand{\thefigure}{S%
\arabic{figure}} \setcounter{figure}{0}
\section*{Supplemental Material for `Construction of Non-Hermitian Parent Hamiltonian from Matrix Product States'}

In this Supplemental Material, we provide more details on matrix product states (MPSs), the verification of the projector, criteria for the existence of nH-PH, the transfer matrix and the metric matrix at the fixed point, calculation of string order, analytical construction of nH-PH for $k=2$, energy spectrum under PBC, and the multi-site iTEBD algorithm.

\section{Introduction on Matrix Product State (MPS)}
Matrix Product State (MPS)~\cite{PerezGarcia2007, Orus2014, Cirac2021}, a member of the tensor network family, has several interesting properties that make them useful in various areas of physics research.
For example, it has been proved that ground states of local, gapped Hamiltonians in 1D spin systems can be efficiently represented by MPS~\cite{Hastings2007}, where the correlation between distant parts of the system decreases exponentially with distance.
These properties make physicists capable of studying interesting phenomena and constructing novel quantum phases in a unified manner~\cite{Chen2010}.
In an MPS representation, global entanglement structure is realized by designing and constructing local tensors, which also facilitates the numerical studies of quantum many-body systems.

For a translation-invariant quantum state, the MPS wavefunction, as shown in Fig.~1(a,b) in the main text, can be written as:
\begin{equation}
    \ket{\Psi} = \sum_{i_1,\dots,i_N} \mathrm{Tr} \left[A^{[i_1]} \dots A^{[i_N]}\right] \ket{i_1,\dots,i_N}
\end{equation}

In this representation, $A^{[i_j]}$ denotes a $D\times D$ matrix that encodes the local degrees of freedom at site $j$ in a one-dimensional chain, where $i_j$ corresponds to the relevant physical indices. 
The dimension of the virtual bond is $D$, which is conventionally denoted as bond dimension.
The trace operation Tr is taken over the virtual indices of the tensors, which connect the first and the last sites of the chain.

In summary, MPS is a powerful mathematical framework that captures the essential features of quantum states in one-dimensional systems. Their connection to entanglement structure makes them a valuable tool for understanding many-body physics, while their computational efficiency makes them attractive for quantum simulation and computation.

\section{Verification of the projector}
Here we diagrammatically verify the projector constructed in Eq.~(6) satisfies the condition Eq.~(3) in the main text for $k = 2$, as shown in Fig.~\ref{VerifyProjector}. 
Projectors for finite $k$ can be verified similarly.

\begin{figure}[htbp]
    \centering
    \includegraphics[width=0.85\columnwidth]{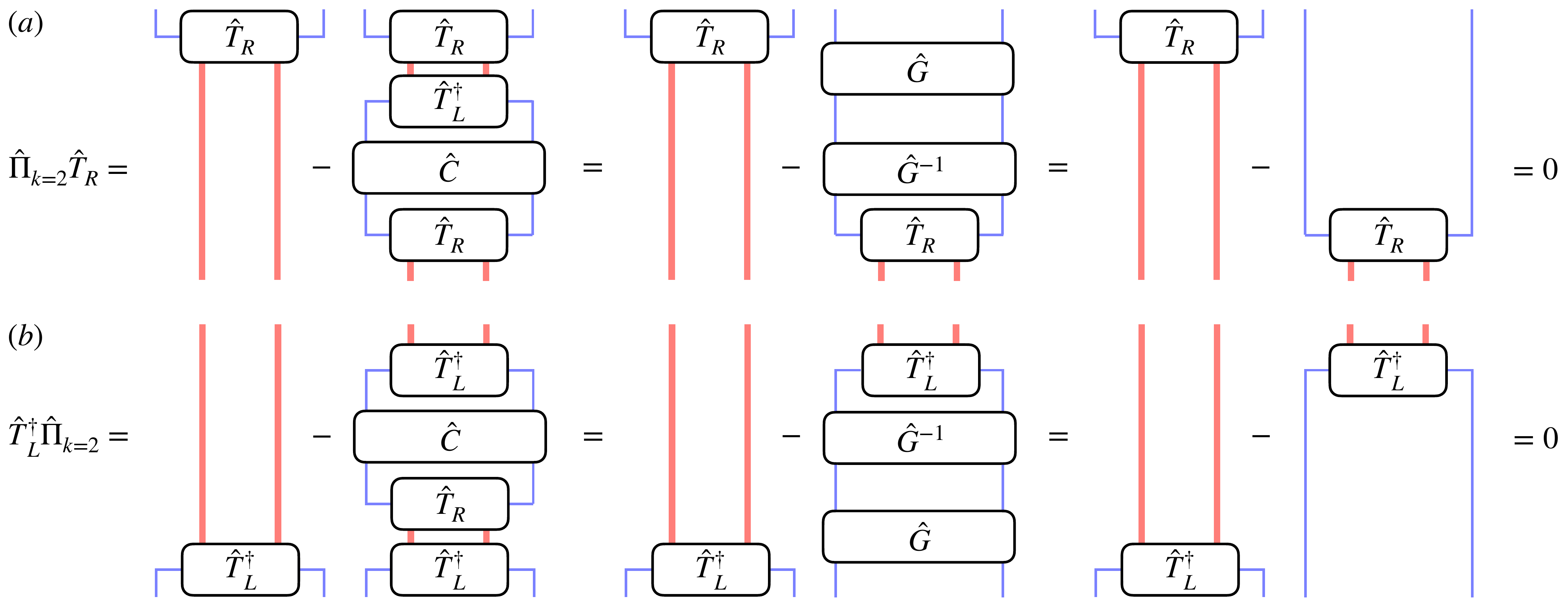}
    \caption{Diagrammatic verification of Eq.~(3) in the main text for $k = 2$.}
    \label{VerifyProjector}
\end{figure}

\section{Criteria for the existence of nH-PH}
In the main text, we have proved that the existence of nH-PH such that the given MPS serves as a zero mode is equivalent to the invertibility of the metric matrix $\mathbf{G}$.
Here we derive another two equivalent criteria, from the perspectives of bi-orthogonalization and direct sum of linear spaces, respectively.

\subsection{Criterion 2}
A central idea in non-Hermitian physics is bi-orthogonalization, which refers to the situation where the left and right eigenstates of a Hamiltonian do not form an orthogonal basis themselves but are orthogonal to each other.
We extend this idea to local Hilbert spaces $\mathcal{H}_R$ and $\mathcal{H}_L$ and show that the ability to find a pair of bi-orthogonal bases on them is equivalent to the ability to construct an nH-PH.
If bi-orthogonalization can be achieved, then there exist invertible transformations $\hat{T}_R^{'} = \ket{\mathbf{p}} \mathbf{T}_R \mathbf{U}_R (\mathbf{r}|$ and $\hat{T}_L^{'} = \ket{\mathbf{p}} \mathbf{T}_L \mathbf{U}_L (\mathbf{l}|$ that satisfy
\begin{align}
    \hat{G}^{'} = \hat{T}_L^{'\dagger} \hat{T}_R^{'} = |\mathbf{l}) \mathbf{U}_L^{\dagger} \mathbf{G} \mathbf{U}_R (\mathbf{r}| = \hat{I}
    \label{Criterion2}
\end{align}
which indicates that $\mathbf{G}$ is invertible.
Conversely, if $\mathbf{G}$ is invertible, we can simply choose $\mathbf{U}_R=\mathbf{G}^{-1}$ and $\mathbf{U}_R=\mathbf{I}$ to satisfy the same relation.
Therefore, the ability to bi-orthogonalize the local Hilbert spaces is equivalent to the ability to construct an nH-PH on them.
This criterion can also be expressed in a basis-independent way.

\subsection{Criterion 3}
Since $\hat{P}$ is not Hermitian in general, its eigenvectors are no longer orthogonal to each other.
To restore the orthogonality, we provide the following lemma and criterion.
\begin{lemma}
    Denote the local Hilbert space of $k$ contracted tensors with physical dimension $d$ as $\mathcal{H}=\mathcal{H}_d^{\otimes k}$.
    For any projector $\hat{P}: \mathcal{H}\rightarrow \mathcal{H}$, $\mathrm{ker}\,\hat{P}^{\dagger} = (\mathrm{im}\,\hat{P})^{\perp}$ and $\mathrm{ker}\,\hat{P} = (\mathrm{im}\,\hat{P}^{\dagger})^{\perp}$.
\end{lemma}
\begin{proof}
    For $\forall \ket{\psi_r} \in \mathrm{ker}\,P$ and $\forall \ket{\psi_l} \in \mathrm{im}\,P^{\dagger}$, we have
    \begin{align}
        \langle \psi_r|\psi_l \rangle
        = \bra{\psi_r} \hat{P}^{\dagger} \ket{\psi_l}
        = \bra{\psi_l} \hat{P} \ket{\psi_r}^{*} 
        = 0.
    \end{align}
    The first equation is a result of $\hat{P}^{\dagger}\ket{\psi_l}=\ket{\psi_l}$. 
    The second equation is obtained by reversing the order of terms with conjugation. 
    The final equation is given by $\hat{P}\ket{\psi_r}=0$.
    Therefore, $\mathrm{ker}\,\hat{P} \in \left(\mathrm{im}\,\hat{P}^{\dagger}\right)^{\perp}$.
    By counting dimensions we know that $\mathrm{ker}\,\hat{P}=\left(\mathrm{im}\,\hat{P}^{\dagger}\right)^{\perp}$.
    Similarly, one can prove that $\mathrm{ker}\,\hat{P}^{\dagger}=\left(\mathrm{im}\,\hat{P}\right)^{\perp}$.
\end{proof}

\begin{figure}[tbp]
    \centering
    \includegraphics[width=0.35\columnwidth]{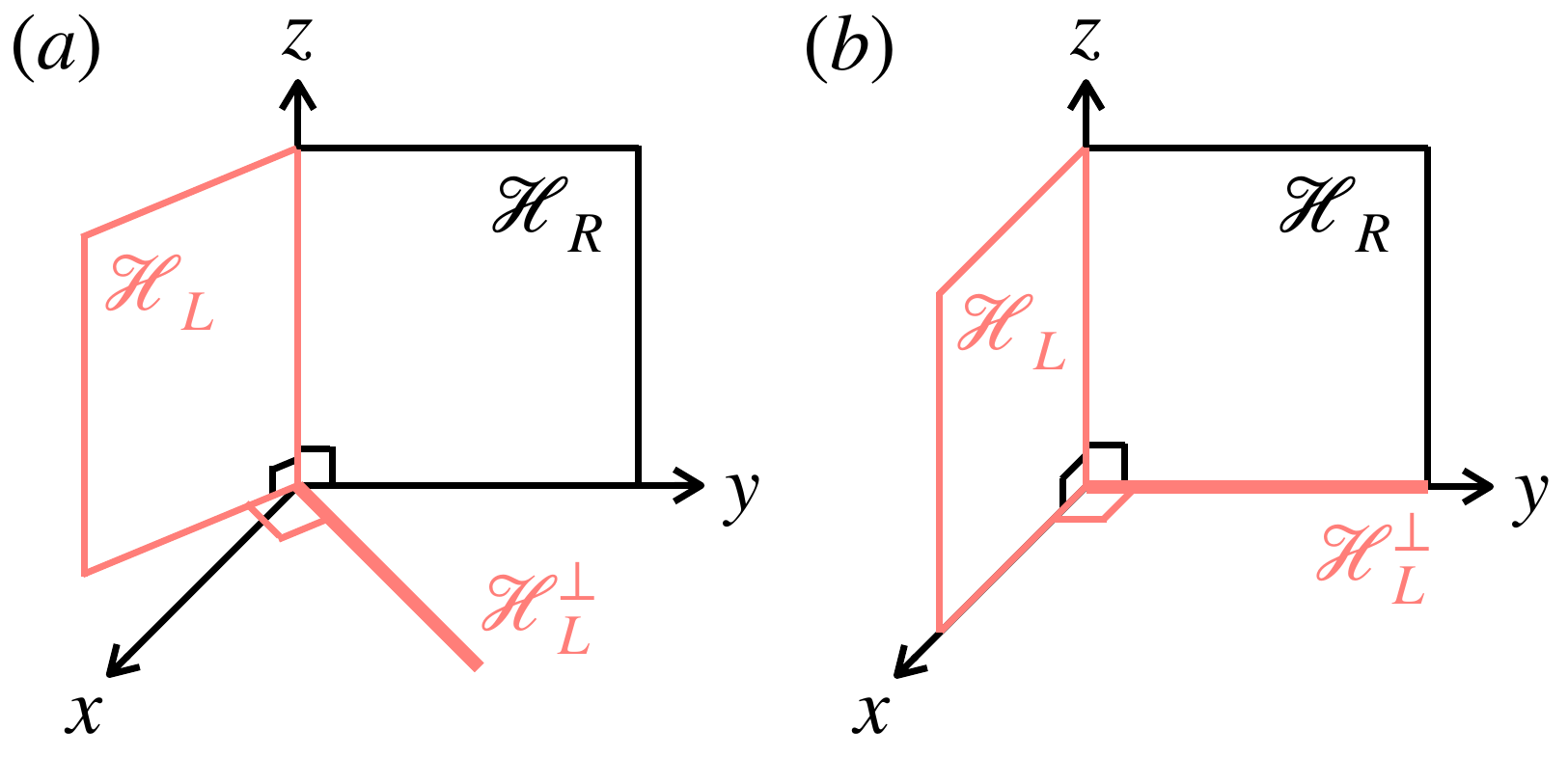}
    \caption{
    An illustrative example of bi-orthogonalization.
    The local Hilbert space $\mathcal{H}=\mathbb{R}^3$.
    The black plane represents $\mathcal{H}_R$ and the red plane represents $\mathcal{H}_L$, with its orthogonal complement $\mathcal{H}_L^{\perp}$ shown as the red line.
    (a) $\mathcal{H}_L$ at a $45^{\circ}$ angle from the $xz$-plane.
    (b) $\mathcal{H}_L$ is the $xz$-plane.
    }
    \label{ToyOrtho}
\end{figure}

In our situation, we expect $\hat{P}$ to satisfy $\mathrm{im}\,\hat{P}=\mathcal{H}_R$ and $\mathrm{im}\,\hat{P}^{\dagger}=\mathcal{H}_L$.
Therefore, by using the lemma, we obtain the following criterion for the existence of such a $\hat{P}$.
\begin{criteria}
    $\mathcal{H}_{R}\oplus \mathcal{H}_{L}^{\perp}$ or $\mathcal{H}_{L}\oplus \mathcal{H}_{R}^{\perp}$ spans the whole local Hilbert space $\mathcal{H}$.
\end{criteria}
\begin{proof}
    If the projector $\hat{P}$ exists, then $\mathrm{ker}\,\hat{P} = (\mathrm{im}\,\hat{P}^{\dagger})^{\perp} = \mathcal{H}_L^{\perp}$, thus $\mathcal{H} = \mathcal{H}_{R}\oplus \mathcal{H}_{L}^{\perp}$.

    From the other side, without loss of generality, we assume that $\mathcal{H} = \mathcal{H}_{R}\oplus \mathcal{H}_{L}^{\perp}$.
    One can construct such a projector $\hat{P}$ that projects onto $\mathcal{H}_R$ satisfying $\mathrm{ker}\,\hat{P} = \mathcal{H}_L^{\perp}$.
    According to the lemma, we have
    \begin{align}\begin{aligned}
        \mathrm{im}\,\hat{P}^{\dagger} &= (\mathrm{ker}\,\hat{P})^{\perp} = \mathcal{H}_L,\\
    \end{aligned}\end{align}
    Therefore, $\hat{P}$ is the desired projector.
\end{proof}

To give an illustrative example, we choose $\mathcal{H}=\mathbb{R}^3$ and $\mathcal{H}_R$ to be the $yz$-plane.
In the first case, $\mathcal{H}_L$ is set at a $45^{\circ}$ angle as shown in Fig.~\ref{ToyOrtho}(a).
It is clear that $\mathcal{H}_L^{\perp}$ and $\mathcal{H}_R$ can span the entire Hilbert space $\mathcal{H}$, therefore an nH-PH can be constructed.
In the second case, $\mathcal{H}_L$ is the $xz$-plane, as depicted in Fig.~\ref{ToyOrtho}(b).
An nH-PH cannot be constructed because the $y$-axis and the $yz$-plane do not span the entire $\mathcal{H}$, even though both $\mathcal{H}_L$ and $\mathcal{H}_R$ have dimension $2$.

\section{The transfer matrix and the metric matrix at the fixed point}\label{FixedPointTransferMatrix}
We analytically calculate the dominant eigenvector of the transfer matrix defined in Eq.~(9) in the main text, from which we construct the transfer matrix and the metric matrix at the fixed point under RG flow.
\begin{align}
    \mathbf{E}_{(\alpha \beta, \alpha' \beta')}^{\infty}(\mu>\frac{1}{3}) = 
    \frac{1}{2}\begin{pmatrix}
        0 \\ -1 \\ 1 \\ 0
    \end{pmatrix}
    \begin{pmatrix}
        0 & -1 & 1 & 0
    \end{pmatrix}
    = \frac{1}{2} \begin{pmatrix}
        0 & 0 & 0 & 0\\
        0 & 1 & -1 & 0\\
        0 & -1 & 1 & 0\\
        0 & 0 & 0 & 0
    \end{pmatrix}.
\end{align}
\begin{align}
    \mathbf{G}_{(\alpha \alpha', \beta \beta')}^{\infty}(\mu>\frac{1}{3}) = 
    \frac{1}{2}\begin{pmatrix}
        0 & 0 & 0 & 1\\
        0 & 0 & -1 & 0\\
        0 & -1 & 0 & 0\\
        1 & 0 & 0 & 0
    \end{pmatrix}.
\end{align}
\begin{align}
    \mathbf{E}_{(\alpha \beta, \alpha' \beta')}^{\infty}(\mu<\frac{1}{3}) = 
    \begin{pmatrix}
        1 \\ 0 \\ 0 \\ 0
    \end{pmatrix}
    \begin{pmatrix}
        1 & 0 & 0 & 0
    \end{pmatrix}
    = \begin{pmatrix}
        1 & 0 & 0 & 0\\
        0 & 0 & 0 & 0\\
        0 & 0 & 0 & 0\\
        0 & 0 & 0 & 0
    \end{pmatrix}.
\end{align}
\begin{align}
    \mathbf{G}_{(\alpha \alpha', \beta \beta')}^{\infty}(\mu<\frac{1}{3}) = 
    \begin{pmatrix}
        1 & 0 & 0 & 0\\
        0 & 0 & 0 & 0\\
        0 & 0 & 0 & 0\\
        0 & 0 & 0 & 0
    \end{pmatrix}.
\end{align}

\section{Calculation of string order}\label{StringOrderAppendix}
\begin{figure}[tbp]
    \centering
    \includegraphics[width=0.9\columnwidth]{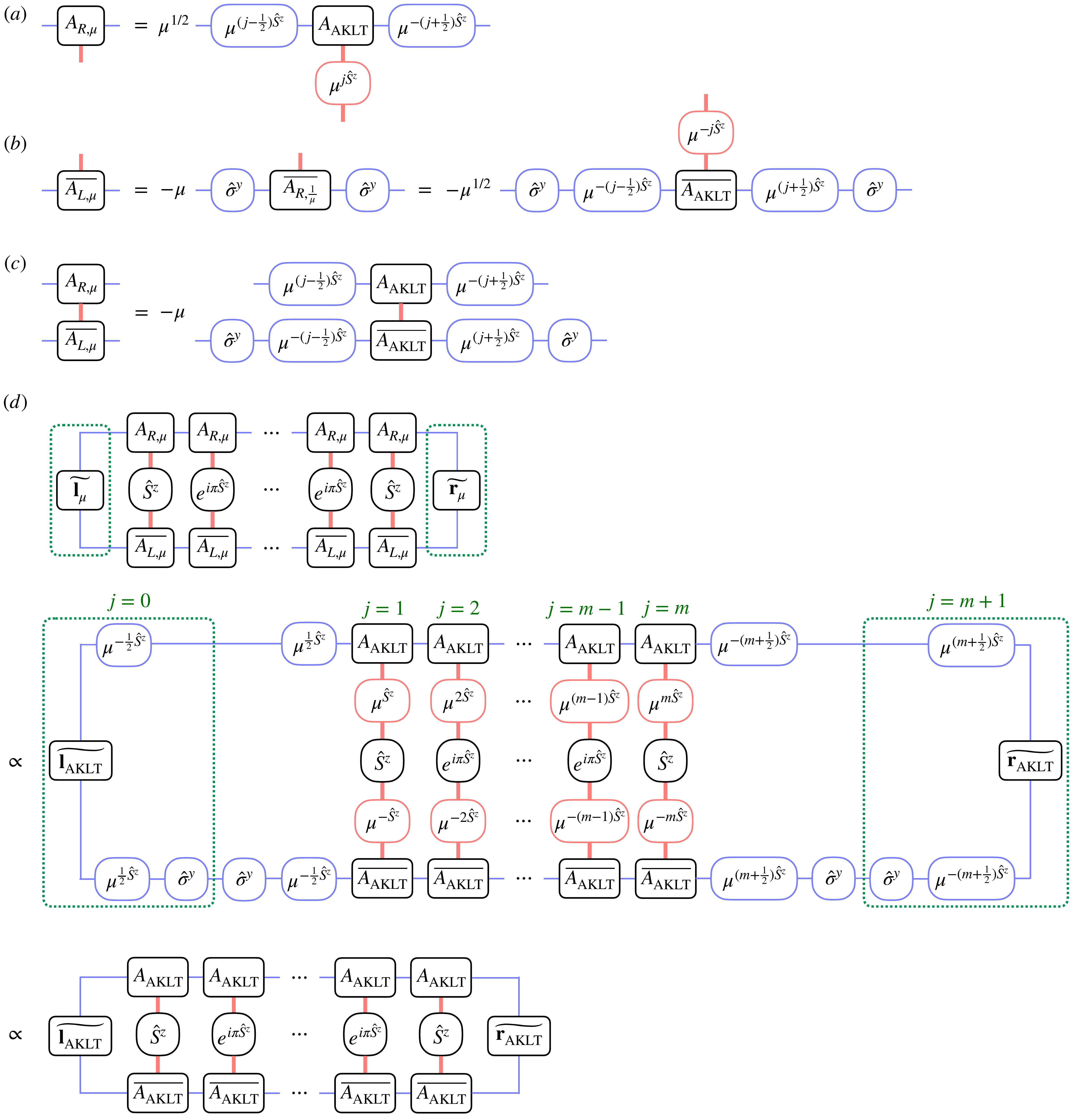}
    \caption{The expectation value of the string order parameter $\braket{\hat{O}_{\rm str}\left(i, j\right)}_{LR}$.
    $\hat{S}^z$ is the angular momentum operator defined on the Hilbert space it acts on.
    (a) The action of $\mu^{j \hat{S}_j^z}$ on $A_{\rm AKLT}$.
    (b) The action of $\mu^{-j \hat{S}_j^z}$ on $\overline{A_{\rm AKLT}}$.
    (c) The relation between transfer matrices.
    (d) $\braket{\hat{O}_{\rm str}\left(i, j\right)}_{LR}$ is the same as that of the AKLT state for $\mu\in(\frac{1}{3}, 3)$.
    }
    \label{StringOrder}
\end{figure}

It was previously shown in~\cite{Matsumoto2020} that $\ket{\Phi_{\mu}}\propto \hat{M}_{\mu} \ket{\Phi_{\rm AKLT}}$ for OBC, where the modification operator is defined as $\hat{M}_{\mu}=\prod_{j=1}^{N} \mu^{j \hat{S}_j^z}$.
However, this cannot be directly applied to the thermodynamic limit as the operator is site-dependent.
In Fig.~\ref{StringOrder}(a)(b), the action of this operator on the on-site tensor of the AKLT state is shown and related to the tensors $A_{\mu}$ and $A_{1/\mu}$ defined in Eq.~(8) in the main text.
In Fig.~\ref{StringOrder}(c), we show the transformation between their transfer matrices. 
From this relation, it can be seen that $\widetilde{\mathbf{l}_\mu}$ ($\widetilde{\mathbf{r}_\mu}$) is proportional to the tensors in the left (right) green dashed square in Fig.~\ref{StringOrder}(d).
Meanwhile, the string order parameter to be calculated $\hat{O}_{\rm str}\left(i, j\right)=\hat{S}_i^z\left(\prod_{k=i+1}^{j-1}e^{i\pi \hat{S}_k^z}\right)\hat{S}_j^z$ commutes with the modification operators $\hat{M}_{\mu}$ and $\hat{M}_{1/\mu}$ in the bulk.
As a consequence, $\braket{\hat{O}_{\rm str}\left(i, j\right)}_{LR}$ is the same as that of the AKLT state for $\mu\in(\frac{1}{3}, 3)$, which is given by $\braket{\hat{O}_{\rm str}\left(i, j\right)}_{LR} = -\frac{4}{9}$, originating from the fact that their fixed point transfer matrices have the same entanglement structure.
Therefore, our non-Hermitian system is in the same quantum phase as the conventional AKLT model in this region.
On the contrary, $\braket{\hat{O}_{\rm str}\left(i, j\right)}_{RR}$ also saturates to a non-zero value within small $m$ for all $\mu$, but the value becomes smaller as $\mu$ deviates away from the AKLT point $\mu = 1$, as shown in Fig.~\ref{Saturate}.

\begin{figure}[tbp]
    \centering
    \includegraphics[width=0.4\linewidth]{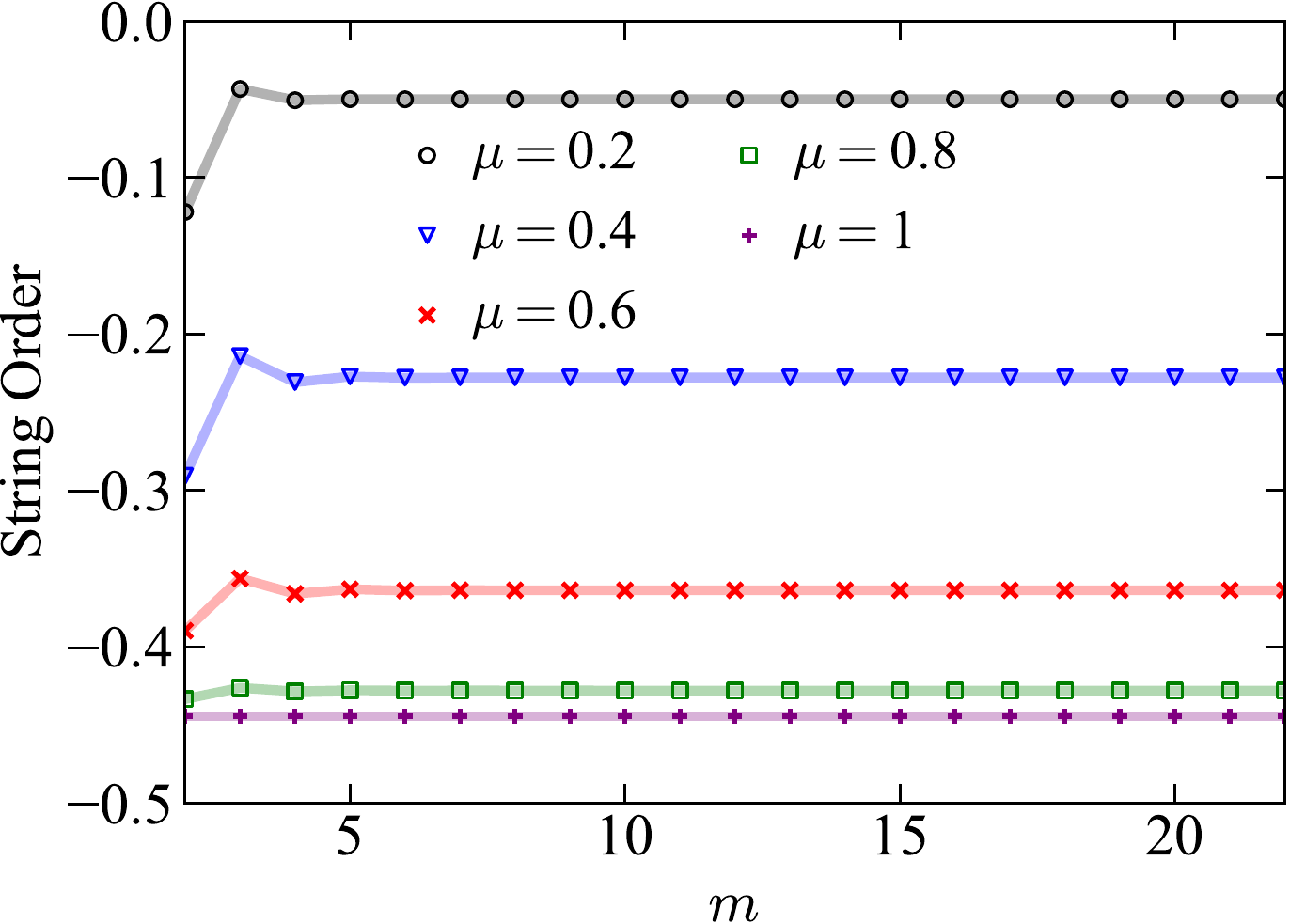}
    \caption{$\braket{O_{\rm str}}_{RR}$ saturate to a non-zero value for all $\mu$, which becomes smaller when away from $\mu=1$.
    }
    \label{Saturate}
\end{figure}
\begin{figure}[tbp]
    \centering
    \includegraphics[width=0.7\linewidth]{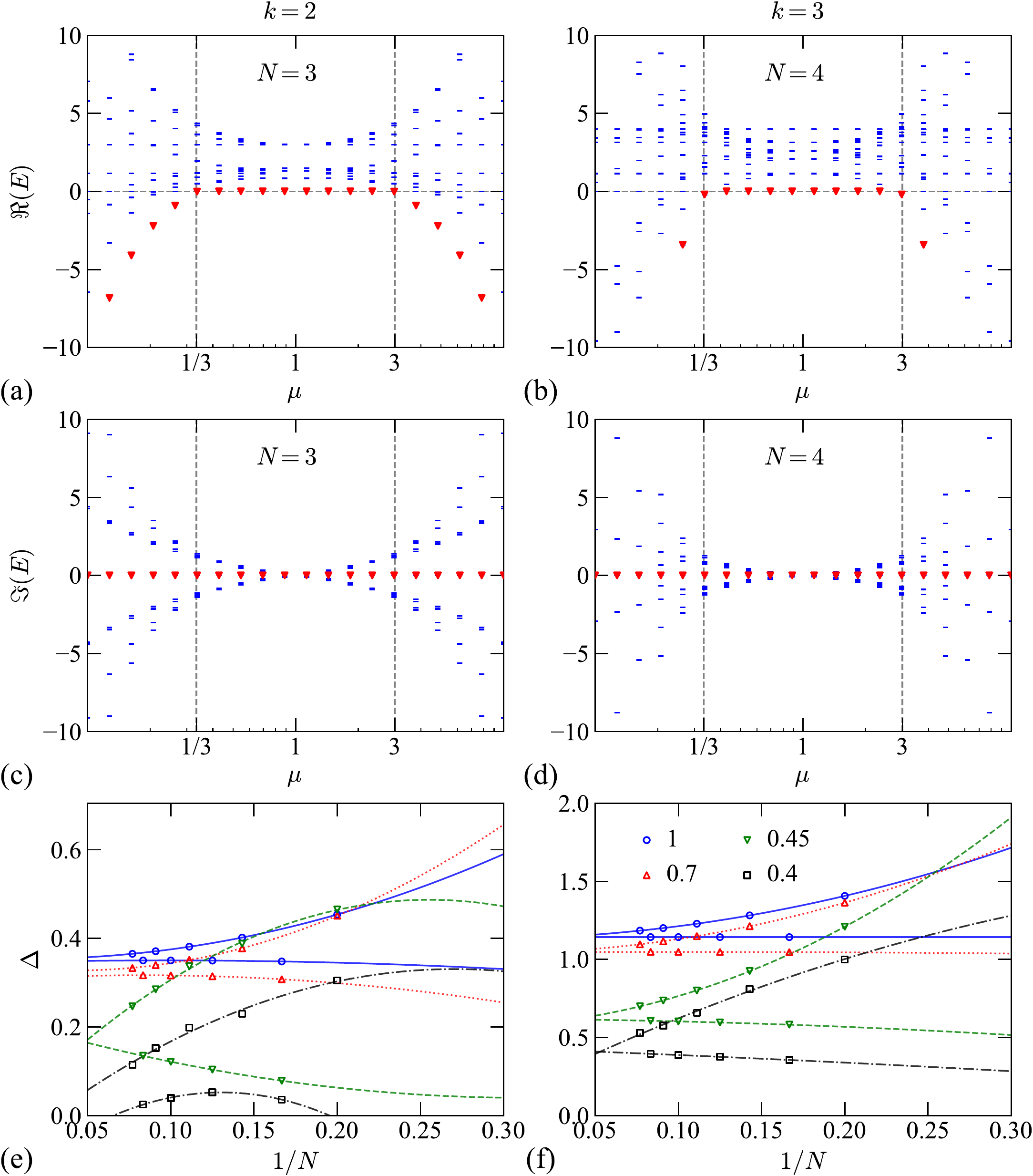}
    \caption{Energy spectrum for finite-size systems under PBC.
        (a)(c) The real and imaginary parts of the spectrum of nH-PH for $k=2$, $N=3$ under PBC.
        Red triangles represent the ground states.
        (b)(d) The same for $k=3$, $N=4$ under PBC.
        (e-f) Finite-size scaling with quadratic functions.
    }
    \label{EDSpectrum}
\end{figure}

\section{Energy spectrum of $\hat{H}_k(\mu)$ for finite-size systems under OBC and PBC}
The entire energy spectrum is independent of $\mu$ under OBC.
In the following, we prove that $\hat{H}_k(\mu)$ and $\hat{H}_k(\mu = 1)$ are related by a similar transformation given by the modification operator $\hat{M}_{\mu}=\prod_{j=1}^{N} \mu^{j \hat{S}_j^z}$, i.e., 
\begin{equation}\begin{aligned}
    \hat{\Pi}_i^{\prime}(\mu) 
    &\equiv \hat{M}_{\mu} \hat{\Pi}_i(\mu = 1) \hat{M}_{\mu}^{-1}\\
    &= \left(\mu^{i \hat{S}_i^z} \otimes \mu^{(i+1) \hat{S}_{i+1}^z}\otimes \cdots\otimes\mu^{(i+k-1) \hat{S}_{i+k-1}^z} \right) \hat{\Pi}_i(\mu = 1) \left(\mu^{-i \hat{S}_i^z} \otimes \mu^{-(i+1) \hat{S}_{i+1}^z}\otimes \cdots\otimes\mu^{-(i+k-1) \hat{S}_{i+k-1}^z}\right) 
\end{aligned}\end{equation}
which is also a projector since $\hat{\Pi}_i^{'2}(\mu)=\hat{\Pi}_i^{'}(\mu)$.
From Fig.~\ref{StringOrder}(a)(b), it can be verified that $\hat{\Pi}_i^{\prime}(\mu) \hat{T}_R(\mu) = \hat{T}_L^{\dagger}(\mu)\hat{\Pi}_i^{\prime}(\mu)= 0$.
Therefore, $\hat{\Pi}_i^{'}(\mu)$ must be the same as $\hat{\Pi}_i(\mu)$ since they project onto the same local Hilbert space $\mathcal{H}_R$, which fully determines the projector.
As a result, the whole Hamiltonian $\hat{H}_{{\rm OBC}}(\mu)$ is related to the AKLT model via a similar transformation $\hat{M}_{\mu}$, so they have the same spectrum.

The spectrum under PBC in Fig.~\ref{EDSpectrum} indicates that $H_k(\mu)$ fails to have $\ket{\Phi_{\mu}}$ as its ground state for $\mu<\frac{1}{3}$ and $\mu>3$, where the ground state energy has a negative real part.
For $N=3,k=2$, the ground state is always unique, while for $N=4,k=3$, the ground state is doubly degenerate in these regions.
It is also worth noting that, due to the finite-size effect, systems at some $\mu$ may be gapped even if they turn out to be gapless in the thermodynamic limit.
Therefore, as shown in Fig.~\ref{EDSpectrum}(e)(f), we use quadratic functions to implement the finite-size scaling to $N\rightarrow\infty$.
For $N=3,k=2$, the systems are gapped at $\mu=1,0.7,0.45$ while gapless at $\mu=0.4$.
For $N=4,k=3$, the systems have a finite gap at all $\mu=1,0.7,0.45,0.4$.
In addition, the gapped region becomes larger as $k$ increases.
These results are consistent with those in Fig.~3 in the main text.

\section{Analytical construction of parent Hamiltonian for $k = 2$}
The spin operators for spin-$1$ are given by
\begin{gather}
\hat{S}^x = \frac{1}{\sqrt{2}}\left(
\begin{array}{ccc}
 0 & 1 & 0 \\
 1 & 0 & 1 \\
 0 & 1 & 0 \\
\end{array}
\right),\,
\hat{S}^y = \frac{1}{\sqrt{2}} \left(
\begin{array}{ccc}
 0 & -i & 0 \\
 i & 0 & -i  \\
 0 & i & 0 \\
\end{array}
\right),\,
\hat{S}^z = \left(
\begin{array}{ccc}
 1 & 0 & 0 \\
 0 & 0 & 0 \\
 0 & 0 & -1 \\
\end{array}
\right),\\
\hat{S}^+ = \hat{S}^x+i\hat{S}^y
= \left(
\begin{array}{ccc}
 0 & \sqrt{2} & 0 \\
 0 & 0 & \sqrt{2} \\
 0 & 0 & 0 \\
\end{array}
\right),\,
\hat{S}^- = \hat{S}^x-i\hat{S}^y
= \left(
\begin{array}{ccc}
 0 & 0 & 0 \\
 \sqrt{2} & 0 & 0 \\
 0 & \sqrt{2} & 0 \\
\end{array}
\right).
\end{gather}
We choose the following set of basis to express the local projector
\begin{align}
    \begin{gathered}
        \hat\lambda_1=\left(\hat{S}^x+\hat{S}^{x z}\right) / 2 \\
        \hat\lambda_2=\left(\hat{S}^x-\hat{S}^{x z}\right) / 2 \\
        \hat\lambda_3=\left(\hat{S}^y+\hat{S}^{y z}\right) / 2, \\
        \hat\lambda_4=\left(\hat{S}^y-\hat{S}^{y z}\right) / 2, \\
        \hat\lambda_5=\hat{S}^{x y} / \sqrt{2}, \\
        \hat\lambda_6=\left(\hat{S}^z+3 \hat{S}^{z 2}\right) / 2\sqrt{2}-\hat{I}/\sqrt{2}, \\
        \hat\lambda_7=(\hat{S}^{x 2}-\hat{S}^{y 2})/\sqrt{2}, \\
        \hat\lambda_8=\left(3 \hat{S}^z-3 \hat{S}^{z 2}+2 \hat{I}\right) / 2 \sqrt{6}, \\
        \hat\lambda_9=\sqrt{\frac{1}{3}} \hat{I}=\frac{1}{2\sqrt{3}}\left(\hat{S}^{x 2}+\hat{S}^{y 2}+\hat{S}^{z 2}\right),
        \end{gathered}
\end{align}
satisfying that $\Tr{[\hat\lambda_i^{\dagger}\hat\lambda_j]}=\delta_{ij}$, where $\hat{S}^{mn}=\hat{S}^m \hat{S}^n+\hat{S}^n\hat{S}^m$.
Under the basis of $\{\hat\lambda_m\}\otimes \{\hat\lambda_n\}$ (i.e., $\hat{O} = \sum_{m,n} \mathbf{O}_{mn}\hat\lambda_m\otimes\hat\lambda_n$), the two-site projector is written as
\begin{align}
    \mathbf{\Pi}(\mu) = 
    \left(\begin{array}{cccc}
    \frac{{\mu}^2+1}{4\mu} & \frac{{\mu}^2+1}{6\mu} & \frac{{\mu}^2-1}{4\mu}i & \frac{{\mu}^2-1}{6\mu}i\\
    \frac{{\mu}^2+1}{6\mu} & \frac{{\mu}^2+1}{4\mu} & \frac{{\mu}^2-1}{6\mu}i & \frac{{\mu}^2-1}{4\mu}i\\
    -\frac{{\mu}^2-1}{4\mu}i & -\frac{{\mu}^2-1}{6\mu}i & \frac{{\mu}^2+1}{4\mu} & \frac{{\mu}^2+1}{6\mu}\\
    -\frac{{\mu}^2-1}{6\mu}i & -\frac{{\mu}^2-1}{4\mu}i & \frac{{\mu}^2+1}{6\mu} & \frac{{\mu}^2+1}{4\mu}\\
    \end{array}
    \right)\oplus
    \left(
    \begin{array}{cccc}
    \frac{{\mu}^4+1}{12{\mu}^2}& 0 & -\frac{{\mu}^4-1}{12{\mu}^2}i & 0\\
    0 & \frac{1}{3} & 0 & \frac{\sqrt{3}}{6}\\
    \frac{{\mu}^4-1}{12{\mu}^2}i & 0 & \frac{{\mu}^4+1}{12{\mu}^2} & 0\\
    0 & \frac{\sqrt{3}}{6} & 0 & \frac{2}{3}
    \end{array}
    \right)\oplus
    \left(\frac{5}{3}\right),
\end{align}
from which we can explicitly construct $\hat{\Pi}$ by local spin operators
\begin{align}
    \begin{aligned}
    \hat{\Pi}(\mu)_{i} 
    &= \frac{5}{12}\left(\frac{\mu}{2} \hat{S}_{i}^{-}\hat{S}_{i+1}^{+} + \frac{1}{2\mu}\hat{S}_{i}^{+}\hat{S}_{i+1}^{-} + \hat{S}_{i}^{z}\hat{S}_{i+1}^{z}\right)
    + \frac{1}{6}\left(\frac{\mu^2}{4}{\hat{S}_{i}^{-2}}{\hat{S}_{i+1}^{+2}}+\frac{1}{4\mu^2}{\hat{S}_{i}^{+2}}{\hat{S}_{i+1}^{-2}}-{\hat{S}_{i}^{z2}}-{\hat{S}_{i+1}^{z2}}\right)\\
    &+ \frac{1}{24}\left(\mu \hat{S}_{i}^{-z}\hat{S}_{i+1}^{+z} + \frac{1}{\mu}\hat{S}_{i}^{+z}\hat{S}_{i+1}^{-z}\right)
    +\frac{1}{4}{\hat{S}_{i}^{z2}}{\hat{S}_{i+1}^{z2}}+\frac{2}{3},
    \end{aligned}
\end{align}
which is just Eq.~(12) in the main text.

\section{Multi-site iTEBD}
Since our local Hamiltonian contains multi-site interactions in general, an algorithm that can handle long-range interactions is needed. 
Here we generalize the modified iTEBD algorithm~\cite{Hastings2009}, which is shown in Fig.~\ref{ModifiediTEBD}.
For a $k$-local Hamiltonian, we consider a $k$-site translation-invariant state with $k$ on-site tensors $T_i$ and $k$ Schmidt weights $G_i$, which are contracted in pairs $X_i = T_iG_i$, as shown in Fig.~\ref{ModifiediTEBD}(a).
We start with the SVD decomposition on $G_kY_k = U_1 G_1^{\prime}Y_{k-1}$ and update $X_1^{\prime}$ by contracting $Y_k$ and $Y_{k-1}^{*}$, which is equivalent to $G_k^{-1}U_1G_1^{\prime}$ but enable us to be free from calculating the inverse of $G_k$, as shown in Fig.~\ref{ModifiediTEBD}(b)(c).
We can do the truncation procedure iteratively until we reach $Y_1$ and end with updating $X_k^{'}=Y_1$.

\begin{figure}[tbp]
    \centering
    \includegraphics[width=1.0\textwidth]{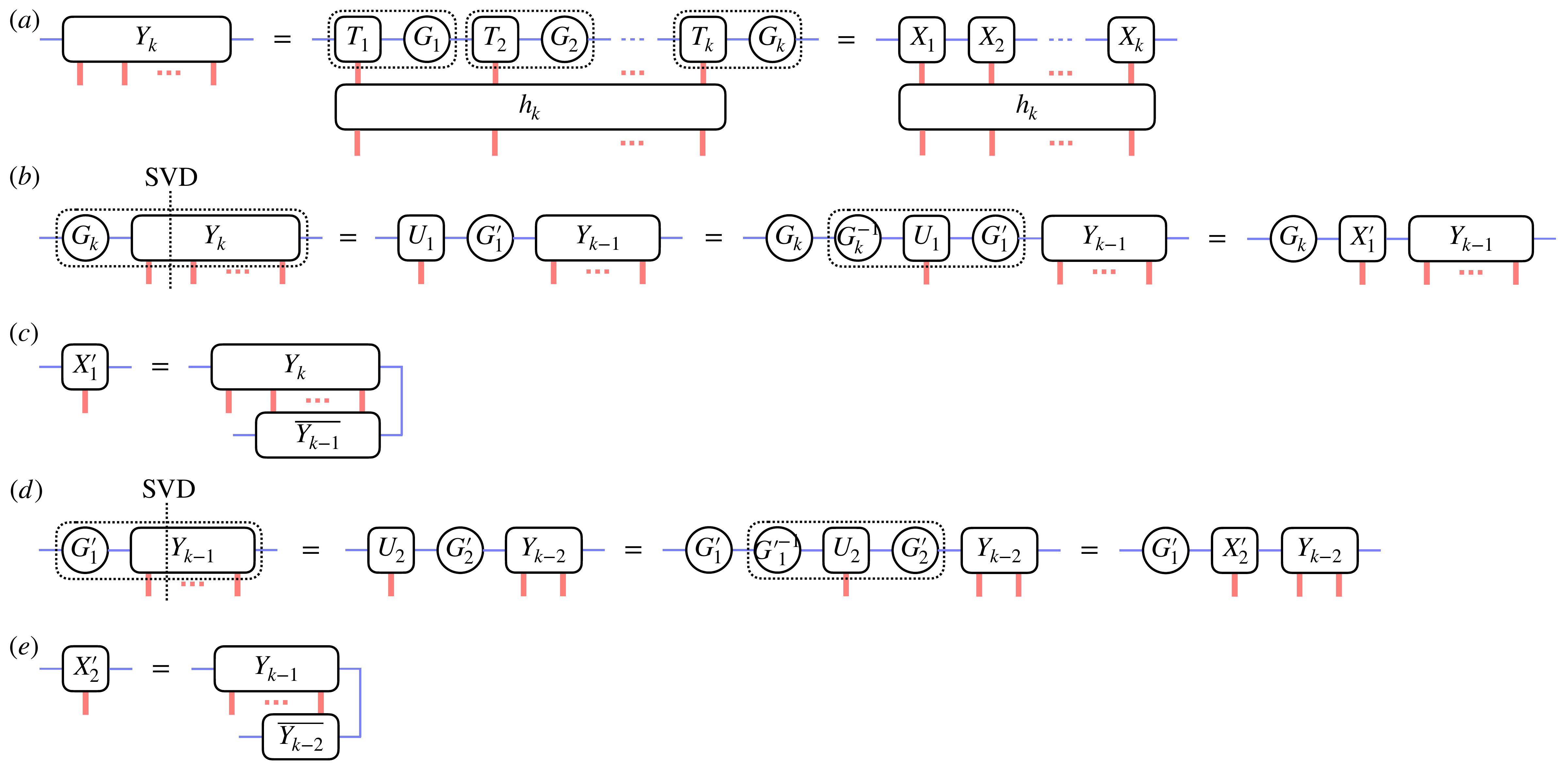}
    \caption{Truncation strategy in the Multi-site iTEBD method.
    (a) Contract $X_i = T_iG_i$.
    (b) Perform SVD decomposition on $G_kY_k = U_1 G_1^{\prime}Y_{k-1}$.
    (c) Update $X_1^{\prime}$ by contracting $Y_k$ and $Y_{k-1}^{*}$.
    (d-e) Iteratively do the same procedure until $Y_1$.}
    \label{ModifiediTEBD}
\end{figure}

We also apply our multi-site iTEBD method to $H_k^{\dagger}$.
The results are shown in Fig.~\ref{iTEBDLeft}, consistent with that in Fig.~3 in the main text.
In numerical simulations for both $H_k(\mu)$ and $H_k(\mu)^{\dagger}$, we adopt the bond dimension $D = 12$, the time step $\Delta\tau=5\times10^{-3}$, and the convergence criterion $e=1\times 10^{-14}$, defined as $e=\sum_{i=1}^{k} \sum_{j=1}^{D}\left[s_{ij}(\tau + \delta\tau) - s_{ij}(\tau)\right]^2$ with $s_{ij}$ being the $j$-th Schmidt weight for site $i$ in the unit cell.

\begin{figure}[tbp]
    \centering
    \includegraphics[width=0.55\columnwidth]{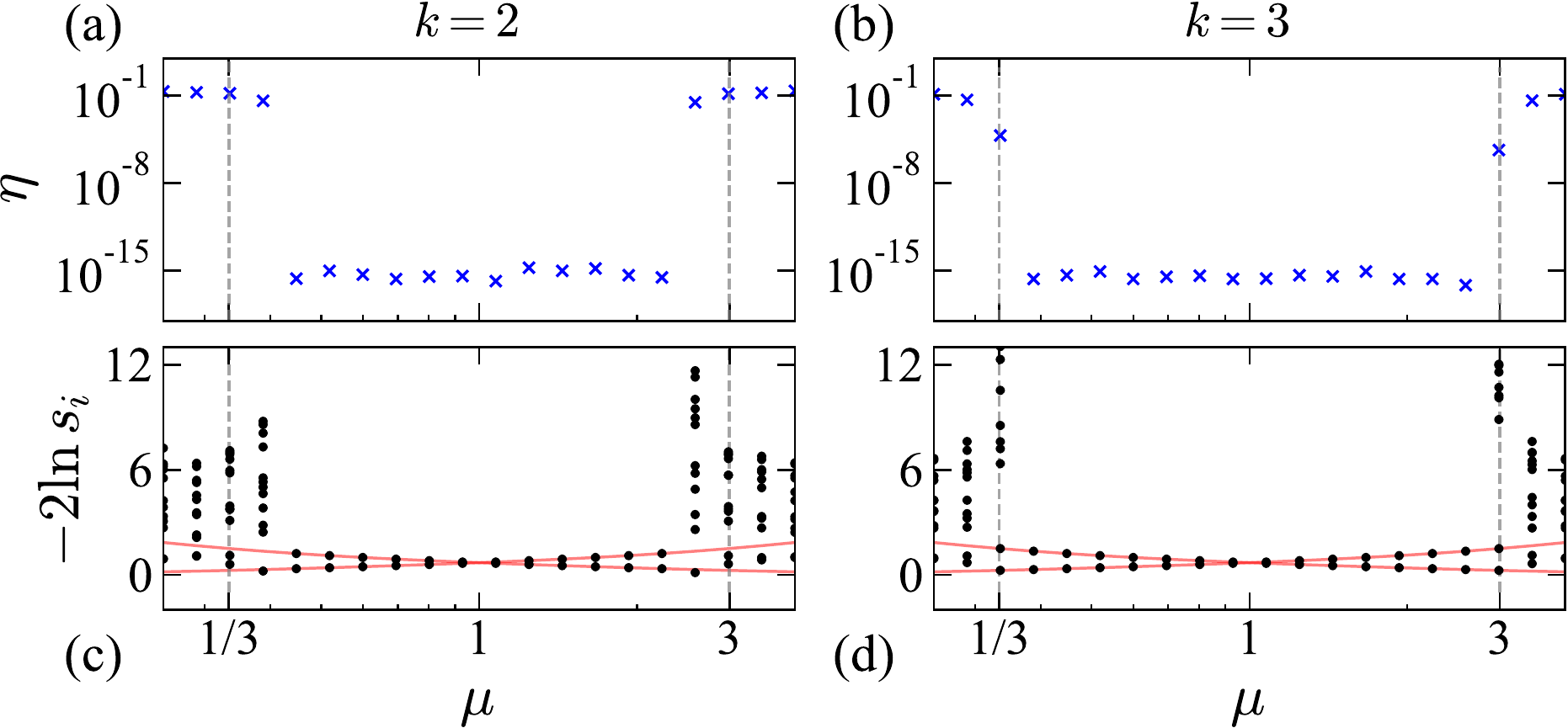}
    \caption{Numerical results for the ground state $\ket{\Psi^{\prime}_{\mu}}$ of $H_k(\mu)^{\dagger}$ calculated from the multi-site iTEBD method with $D = 12$ for $k = 2$ and $k = 3$.
    (a-b) Infidelity between $\ket{\Psi^{\prime}_{\mu}}$ and $\ket{\Phi_{1/\mu}}$.
    (c-d) Entanglement spectrum of $\ket{\Psi^{\prime}_{\mu}}$ (black dots) and $\ket{\Phi_{1/\mu}}$ (red lines).
    }
    \label{iTEBDLeft}
\end{figure}
\end{document}